\documentclass[twoside,leqno,twocolumn]{article}

\usepackage[letterpaper]{geometry}
\usepackage{amsmath, amssymb, amsfonts, bm, cite}

\usepackage{ltexpprt}

\newcommand{\RR}{{\mathbb R}}

\newcommand{\cI}{\mathcal{I}}
\newcommand{\cM}{\mathcal{M}}

\newcommand{\cC}{\mathcal{C}}
\newcommand{\cX}{\mathcal{X}}

\newcommand{\E}{{\mathbb E}}

\newcommand{\vect}[1]{\mathbf{#1}}
\newcommand{\bx}{\vect{x}}

\newcommand{\bw}{\vect{w}}
\newcommand{\by}{\vect{y}}
\newcommand{\bz}{\vect{z}}

\def\b1{{\bf 1}}

\begin{document}

\title{\Large Estimating the Nash Social Welfare \\
for coverage and other submodular valuations}

\author{Wenzheng Li\thanks{Stanford University} \and Jan Vondr\'{a}k\thanks{Stanford University}}

\date{}

\maketitle


\fancyfoot[R]{\scriptsize{Copyright \textcopyright\ 20XX by SIAM\\
Unauthorized reproduction of this article is prohibited}}





\begin{abstract} \small\baselineskip=9pt We study the Nash Social Welfare problem: Given $n$ agents with valuation functions $v_i:2^{[m]} \rightarrow \RR_+$, partition $[m]$ into $S_1,\ldots,S_n$ so as to maximize $(\prod_{i=1}^{n} v_i(S_i))^{1/n}$. The problem has been shown to admit a constant-factor approximation for additive, budget-additive, and piecewise linear concave separable valuations; the case of submodular valuations is open. 

We provide a $\frac{1}{e} (1-\frac{1}{e})^2$-approximation of the {\em optimal value} for several classes of submodular valuations: coverage, sums of matroid rank functions, and certain matching-based valuations.
\end{abstract}

\section{Introduction}

Nash Social Welfare is an optimization problem in the following form:

\paragraph{Nash Social Welfare (NSW).}
Given $m$ indivisible items and $n$ agents with valuation functions 
$v_i:\{0,1\}^m \rightarrow \RR_+$,
we want to allocate the items to the agents, that is find ${\bf x}\in\{0,1\}^{n\times m}$ such that for each $j$, $\sum_{i=1}^{n} x_{ij} \leq 1$, in order to maximize the {\em geometric average} of the valuations,
$$ \nu(\bx) = \left( \prod_{i=1}^{n} v_i(\bx_i) \right)^{1/n}.$$

The notion of Nash Social Welfare goes back to John Nash's work \cite{Nash50} on bargaining in the 1950s. It can be viewed as a compromise between Maximum Social Welfare (maximizing the summation $\sum_{i=1}^{n} v_i(\bx_i)$, which does not take fairness into account), and Max-Min Welfare (maximizing $\min_{1 \leq i \leq n} v_i(\bx_i)$, which is focusing on the least satisfied agent ignores the possible additional benefits to others). Nash Social Welfare also came up independently in the context of competitive equilibria with equal incomes \cite{Var74} and proportional fairness in networking \cite{Kel97}. An interesting feature of Nash Social Welfare is that the optimum is invariant under scaling of the valuations $v_i$ by independent factors $\lambda_i$; i.e., each agent can express their preference in a ``different currency" and this does not affect the problem.

The difficulty of the problem naturally depends on what class of valuations $v_i$ we consider. Unlike the (additive) Social Welfare Maximization problem, the Nash Social Welfare problem is non-trivial even in the case where the $v_i$'s are {\em additive}, that is $v_i(\bx_i) = \sum_{j=1}^{m} w_{ij} x_{ij}$. It is NP-hard even in the case of 2 agents with identical additive valuations (by a reduction from the Subset Sum problem), and APX-hard for multiple agents \cite{Lee17}. Constant-factor approximations for the additive case were discovered only recently, in remarkable works by Cole and Gatskelis \cite{CG18}, and subsequently via a very different algorithm by Anari et al. \cite{AGSS17}. Inspired by the exciting breakthrough, a series of followup work have also been developed along this line \cite{barman2018finding, cole2017convex}.


A natural question is whether a constant-factor approximation extends to broader classes of valuations, for example submodular valuations (where a constant factor is known for Maximum Social Welfare \cite{FNW78,Von08}). This is unknown at the moment. Some progress has been made for classes of valuations beyond additive ones: a constant factor for concave piece-wise linear separable utilities \cite{AMGV18}, and for budget-additive valuations \cite{GHM19,CCG18}.

The problem seems particularly challenging for general submodular valuations. Only very recently, an $O(n \log n)$-approximation has been designed for submodular valuations in \cite{GKK20}. (The authors also present a $(1-1/e-\epsilon)$-approximation in the case where the number of agents is constant, which follows from earlier work on multiobjective submodular optimization.) 

\subsection{Our results}

We consider the Nash Social Welfare optimization problem for several subclasses of submodular valuations. Our main results are constant-factor approximations {\em of the optimal value} for several such classes that we discuss below. We stress that {\em we do not know} how to find an allocation of corresponding value in polynomial time, which is a situation reminiscent of certain variants of max-min allocation. We discuss this in more detail below.

The classes of submodular valuations that we consider are:

\begin{itemize}
\item Matroid rank functions: Given a matroid $\cM = ([m], \cI)$, its rank function is $$r_{\cM}(S) = \max \{|I|: I \subseteq S, I \in \cI\}.$$ More generally, for $\bw \in \RR_+^m$, a weighted matroid rank function is $$r_{\cM,\bw}(S) = \max \{ \sum_{j \in I} w_j: I \subseteq S, I \in \cI \}.$$

We obtain a $\frac{1}{e} (1-\frac{1}{e})^2$-approximation for this class.

\item The cone generated by matroid rank functions: $\cC =$ 
$$  \left\{ \sum_{\mbox{\tiny matroid } \cM} \alpha_\cM r_{\cM}: \alpha_\cM \geq 0 \mbox{ for each matroid } \cM \right\}.$$
This includes weighted matroid rank functions as a subclass. The cone also contains coverage valuations, $v(S) = |\bigcup_{j \in S} A_j|$ for a set system $(A_j)_{j=1}^{m}$.

We also obtain a $\frac{1}{e} (1-\frac{1}{e})^2$-approximation for this class.

\item Bipartite matching with a matroid constraint: Let $G = (L, R, E)$ be a bipartite graph with non-negative weight $w_{j, k}$ on the edge $(j, k) \in E$, vertex set on the left-hand side $L = [n]$ and a matroid $\mathcal{M}$ defined on the vertex set on the right side $R$.
For $S \subseteq L$, let $v(S)$ be the maximum weight of a matching such that vertices matched on the left are a subset of $S$, and the vertices matched on the right are an independent set in $\mathcal{M}$.

It is known that these valuations satisfy the property of {\em gross substitutes} (a subclass of submodular functions for which the welfare maximization problem can be solved optimally). In fact there is apparently no known example of a gross-substitutes function which is not in this form. \footnote{Based on Kazuo Murota's problem set \cite{Murota15} and personal communication with Renato Paes Leme's tutorial \cite{Leme18}}

We still obtain a factor of $\frac{1}{e} (1-\frac{1}{e})^2$ for this class, as well as the cone generated by such functions.

\item More generally, we can handle classes of (non-submodular) valuations that support a certain form of {\em contention resolution}. As an example, we show that it is possible to obtain a $\Theta(1/k)$-approximation for valuations defined by a $k$-dimensional matching problem with matroid constraints in Appendix~\ref{appen:high-dimension}. Similarly we can also handle the cone generated by such functions. 

\end{itemize}

\subsection{Techniques and discussion}

Our approach is based on a convex programming relaxation from \cite{AGSS17}, which relies on properties of stable polynomials and was inspired by Gurvits' proof of the van der Waerden conjecture for the permanent of doubly stochastic matrices. We extend this relaxation to the classes of submodular functions that we consider, and prove that its integrality gap is bounded by $\frac{1}{e} (1-\frac{1}{e})^2$. This factor arises from two sources: The first one is a generalization of Gurvits' lemma \cite{gurvits2005complexity}, developed in \cite{AGSS17}, which yields a factor of $\frac{1}{e}$. The second ingredient is the use of a {\em contention resolution scheme} for matroids, which we use twice to deal with the presence of matroid rank functions as well as the assignment constraint of the problem. Combining these ingredients via the FKG inequality yields a factor of $\frac{1}{e} (1- \frac{1}{e})^2$.

Technically, we prove that a certain rounding technique achieves a $c^n$ approximation factor in terms of the expected {\em product of valuations} $\E[\prod_{i=1}^{n} v_i(\bx_i) ]$. This proves that there exists a solution whose value in terms of the {\em geometric average} $(\prod_{i=1}^{n} v_i(\bx_i) )^{1/n}$ is at least $c \cdot OPT$. 
However, the rounding method is randomized and a solution of value $c \cdot OPT$ might appear with an exponentially small probability so the expect running time would be exponential. 

A similar issue already appears in \cite{AGSS17}, where the authors resolve this by a method of conditional expectations. The estimator of conditional expectations turns out to be a matching-counting problem which is known to admit an FPTAS by Markov Chain Monte Carlo methods. 
In our case, an analogous approach leads to counting problems which we currently don't know how to resolve.
In particular, in the case of coverage valuations, we would need to resolve the following counting problem.




\paragraph{Open problem: Counting of constrained mappings.}
Given sets $A, B$, where $B$ is partitioned into blocks $B_1,\ldots,B_n$, and weights $w_{ij} \geq 0$ for all $i \in A, j \in B$, for each choice of $e_1 \in B_1, e_2 \in B_2, \ldots, e_n \in B_n$ let ${\cal S}(e_1,\ldots,e_n)$ be the set of all mappings $\sigma:A \rightarrow B$ such that $e_1,\ldots,e_n \in \sigma[A]$. Estimate the quantity
$$ \sum_{e_1 \in B_1,\ldots,e_n \in B_n} \sum_{\sigma \in {\cal S}(e_1,\ldots,e_n)} \prod_{i \in A} w_{i,\sigma(i)}.$$

Finally, we want to remark that the situation is somewhat reminiscent of the Santa Claus problem, for which the first constant-factor integrality gap bound was proved by Uriel Fiege in \cite{Feige08}. This proof relied on the use of the Lov\'asz Local Lemma, which was later made algorithmically efficient in \cite{HSS11}.  While our analysis does not use the Lov\'asz Local Lemma, there are certain similarities here. In some special cases, the rounding method essentially attempts to build a perfect matching by selecting a random match for each vertex. The probability of success is obviously exponentially small.

\section{Matroid rank functions} \label{sec:matroid}

We present first our algorithm for (weighted) matroid rank functions, which is technically the simplest case but it already illustrates the basic ideas of our approach. Further generalizations require some technicalities but mostly follow along the same lines.

Let the valuation of player $i$ be a weighted matroid rank function,
$$ v_i(\bx_i) = \max \{ \sum_{j \in I} w_{ij}: I \in \cI_i, \b1_I \leq \bx_i \} $$
where $\cM_i = ([m], \cI_i)$ is a matroid and $w_{ij} \geq 0$ for all $i \in [n], j \in [m]$.

Let $P(\cM_i)$ denote the associated matroid polytope,
$$ P(\cM_i) = \mbox{conv}( \{ \b1_I: I \in \cI_i \}).$$
It is known that the weighted rank function can be equivalently written as
$$ v_i(S) = \max \{ \sum_{j \in S} w_{ij} x_{ij}: \bx_i \in P(\cM_i) \}.$$
With this in mind, we suggest the following relaxation (an extension of the convex/concave program from \cite{AGSS17}). In the rest of the paper, we will also use similar notations as \cite{AGSS17}), denoting $\prod_{i\in S}y_i$ by $y^S$ and $\{S\subseteq [m] \big| |S| = n\}$ by $\binom{[m]}{n}$.

\subsection{Convex relaxation for matroid rank functions}

\begin{align*}
(P_1) \\
&\max_{\bx \in \RR^{n \times m}}  \inf_{\by \in \RR_+^m: y^S\ge1,\forall S\in\binom{[m]}{n}} \prod_{i=1}^n&\left(\sum_{j=1}^{m} w_{ij} x_{ij} y_j \right), \\
& s.t. ~~~\bx_i \in P(\cM_i) & \forall i \in [n] \\
& ~~~~~~ \sum_{i=1}^{n} x_{ij}\le 1 & \forall j \in [m] \\
\end{align*}

This is indeed a relaxation, as the next lemma shows .

\begin{lemma}
The optimal solution of the ($P_1$) relaxation is at least the optimal solution of the Nash welfare maximization problem.
\end{lemma}

\begin{proof}
Suppose that $\bx \in \{0,1\}^{n \times m}$ is the optimal allocation.
Since the valuation for player $i$ is a weighted matroid rank function for $\cM_i$, we can assume without loss of generality that the set allocated to player $i$ is independent in $\cM_i$ (if necessary, remove some items to retain only the independent set that defines the value of $v_i(\bx_i)$). Then we have $\bx_i \in P(\cM_i)$ and $\sum_{i=1}^{n} x_{ij} \leq 1$. Also, for every $\by \in \RR_+^m$ such that $y^S \geq 1$  for all $S \in {[m] \choose n}$, we have
\begin{align*}
    \prod_{i=1}^{n} \left( \sum_{j=1}^{m} w_{ij} x_{ij} y_j \right) & = \sum_{j_1,\ldots,j_n=1}^{m} \prod_{i=1}^{n} (w_{ij_i} x_{ij_i} y_{j_i}) \\
    & = \sum_{\mbox{\tiny distinct } j_1,\ldots,j_n} \prod_{i=1}^{n} y_{j_i} \prod_{i=1}^{n} (w_{ij_i} x_{ij_i}) \\
    & \geq \sum_{\mbox{\tiny distinct } j_1,\ldots,j_n} \prod_{i=1}^{n} (w_{ij_i} x_{ij_i}) \\
    & = \prod_{i=1}^{n} \left( \sum_{j=1}^{n} w_{ij} x_{ij}\right)
\end{align*}  

where the inequality holds because of the constraint $y^S \geq 1$, and the equality between summations over all $n-$tuples and distinct $n$-tuples holds because we cannot have $x_{ij} = x_{i' j'} = 1$ for $i \neq i'$ and $j = j'$.
Therefore, the optimum of the relaxation is at least the integer optimum.
\end{proof}

Next, we use the following lemma from \cite{AGSS17}, which explains why the relaxation can be solved efficiently.

\begin{lemma}
\label{lem:concave-convex}
The logarithm of the objective function,
$$ \sum_{i=1}^{n} \log \sum_{j=1}^{m} w_{ij} x_{ij} y_j $$
is concave in $x_{ij}$ and convex in $\log y_j$.
\end{lemma}

The added ingredient here is the matroid polytope constraint; however as is well known, matroid polytopes admit efficient separation oracles, and hence our relaxation can still be solved (within an arbitrarily small error) by standard techniques just like in \cite{AGSS17}.

\subsection{Rounding for matroid rank functions}

The main remaining question is how to round a fractional solution of $(P_1)$. In \cite{AGSS17}, items are allocated to players independently with probabilities $x_{ij}$, which is a natural choice. However, here it is not clear if this works due to the nonlinearity of the valuation functions. We need one additional ingredient, which is {\em contention resolution} \cite{CVZ14}.

\begin{Definition}
A $(b,c)$-balanced contention resolution for an independence system $\cM = ([m], \cI)$ (e.g. a matroid) is a procedure that for every $\bx \in b \cdot P(\cM)$ and $A \subseteq [m]$ returns a (random) set $\pi_\bx(A) \subseteq A$ such that 
\begin{itemize}
    \item $\pi_\bx(A) \in \cI$ with probability $1$, and
    \item if $R(\bx)$ is a random set where each element $j$ appears independently with probability $x_j$, then
    $$ \Pr[j \in \pi_\bx(R(\bx)) \mid j \in R(\bx)] \geq c.$$
\end{itemize}
The scheme is said to be monotone if $\Pr[i \in \pi_\bx(A_1)] \ge \Pr[i \in \pi_\bx(A_2)]$
whenever $i \in A_1 \subseteq A_2$.
\end{Definition}

Contention resolution schemes are known now for various types of constraints. We mention the following result, most relevant to our work here \cite{CVZ14}.

\begin{theorem}
\label{thm:CRS-matroid}
For any matroid and any $b \in (0,1]$, there is a monotone $(b, \frac{1-e^{-b}}{b})$-balanced contention resolution scheme.
\end{theorem}

We are going to use contention resolution to make sure that every player receives a set independent in their respective matroid. We need to proceed carefully, since we need to combine contention resolution with the analysis of the product of valuation functions, so any dependencies between players (as in the rounding procedure of \cite{AGSS17}) are potentially dangerous. We design the rounding procedure as follows.

\

\paragraph{Rounding Procedure 1 \label{rounding:CR-matroid}: Rounding for matroid rank functions.}
Given a fractional solution $\bx \in \RR^{n \times m}$,
\begin{itemize}
    \item For each $(i,j) \in [n] \times [m]$ independently, set $X_{ij} = 1$ with probability $x_{ij}$ and $X_{ij} = 0$ otherwise.
    \item For each $j \in [m]$ independently, apply contention resolution in the uniform rank-1 matroid on $[n]$ to the set $P_j = \{i: X_{ij} = 1\}$, to obtain a singleton $\{p_j\}$ (the player tentatively receiving item $j$).
    \item For each $i \in [n]$ independently, apply contention resolution in matroid $\cM_i$ to the set $S_i = \{j: X_{ij} = 1\}$, to obtain an independent set $I_i \in \cI_i$ (the set tentatively allocated to player $i$).
    \item Allocate item $j$ to player $p_j$, if $j \in I_i$.
\end{itemize}

We remark that the 3rd bullet point and the condition $j \in I_i$ are not crucial for the algorithm to work (and in fact the algorithm might potentially lose some value by applying it). However, it ensures that each player receives an independent set, which will be useful in the analysis. 

We now proceed to analyze the algorithm. Let us set some notation. Let $Y_{ij}$ denote the indicator variable for the event that $p_j = i$ (i.e., the $(i,j)$ pair survives the contention resolution in the rank-1 matroid). Let $Z_{ij}$ denote the indicator variable for the event that $j \in S_i$  (i.e., the $(i,j)$ pair survives the contention resolution in matroid $\cM_i$). Observe that $Y_{ij} Z_{ij}$  is the indicator of the event that player $i$ actually receives item $j$. 

\begin{lemma}
\label{lem:distinct-YZ}
Let $j_1,\ldots,j_n \in [m]$ be distinct items. Then
$$ \E\left[\prod_{i=1}^{n} (Y_{i j_i} Z_{i j_i})\right] \geq (1-1/e)^{2n} \prod_{i=1}^{n} x_{i j_i}.$$
\end{lemma}

\begin{proof}
Let $\cX$ denote the event that $X_{i j_i} = 1$ for every $i \in [n]$. This is a necessary condition for $\prod_{i=1}^{n} (Y_{i j_i} Z_{i j_i}) = 1$ to happen, and by independence, $\Pr[\cX] = \prod_{i=1}^{n} x_{i j_i}$. Hence our goal is to prove that $\E\left[\prod_{i=1}^{n} (Y_{i j_i} Z_{i j_i}) \mid \cX \right] \geq (1-1/e)^{2n}$.

Note that for each $i'$, $X_{i' j_{i'}}$ is the only variable among $(X_{i j_i})_{i=1}^{n}$ that participates in the contention resolution scheme in $\cM_{i'}$, and similarly it is the only variable that participates in the contention resolution scheme on $\{(i,j_{i'}): 1 \leq i \leq n\}$. Therefore, for the purposes of these contention resolution schemes, conditioning in $\cX$ is equivalent to conditioning on $X_{i' j_{i'}} = 1$. By the properties of contention resolution in matroids, we have $\Pr[Y_{i j_i} = 1 \mid \cX] \geq 1-1/e$ and $\Pr[Z_{i j_i} = 1 \mid \cX] \geq 1-1/e$. However, there are dependencies between these events for different values of $i$ (since conditioning on $Y_{ij}, Z_{ij}$ gives information about other elements participating in the contention resolution schemes), so the lemma doesn't follow immediately.

We need the additional property that our contention resolution is {\em monotone}, that is elements are less likely to survive when selected from a larger set. More formally, for $P \subseteq [n]$ such that $i \in P$, and $S \subseteq [m]$ such that $j \in S$, let 
$$\pi_{ij}(P) = \Pr[Y_{ij} = 1 \mid \cX \ \& \  P_j = P],$$
$$\sigma_{ij}(S) = \Pr[Z_{ij} = 1 \mid \cX \ \& \  S_i = S].$$
From Theorem~\ref{thm:CRS-matroid}, $\pi_{ij}(P)$ is a non-increasing function of $P$ (among sets $P$ such that $i \in P$), and $\sigma_{ij}(S)$ is a non-increasing function of $S$ (among sets $S$ such that $j \in S$).
Also, note that once the variables $X_{ij}$ are fixed, the sets $P_j, S_i$ are determined, and each contention resolution procedure proceeds independently. Hence, we have
$$ \E\left[\prod_{i=1}^{n} Y_{i j_i} Z_{i j_i} \mid \cX \right] = \E\left[\prod_{i=1}^{n} \pi_{i j_i}(P_{j_i}) \sigma_{i j_i}(S_i) \mid \cX \right].$$

Now we invoke the FKG inequality: For non-decreasing functions $\pi(\omega), \sigma(\omega)$ on a product space $\Omega$ with a product measure, $\E_\Omega[\pi(\omega) \sigma(\omega)] \geq \E_\Omega[\pi(\omega)] \E_\Omega[\sigma(\omega)]$. Here, the product space is the space of random variables $X_{ij}$ conditioned on $\cX$ (which still keeps the remaining random varables $\{X_{ij}, j \neq j_i\}$ independent). Hence, by applying the FKG inequality repeatedly,
\begin{align*} &\E\left[\prod_{i=1}^{n} \pi_{i j_i}(P_{j_i}) \sigma_{i j_i}(S_i) \mid \cX \right] \\
\geq &\prod_{i=1}^{n} \left( \E\left[\pi_{i j_i}(P_{j_i}) \mid \cX \right] \cdot \E\left[\sigma_{i j_i}(S_i) \mid \cX \right] \right).
\end{align*}

Finally, from Theorem~\ref{thm:CRS-matroid}, we get
$$ \E_{P_j}[ \pi_{ij}(P_j) \mid \cX] = \Pr[Y_{ij} = 1 \mid \cX] \geq 1 - {1}/{e},$$
$$ \E_{S_i}[ \sigma_{ij}(S_i) \mid \cX] = \Pr[Z_{ij} = 1 \mid \cX] \geq 1 - {1}/{e}.$$
This completes the proof that $\E\left[\prod_{i=1}^{n} (Y_{i j_i} Z_{i j_i}) \mid \cX \right] \geq (1-1/e)^{2n}$.
\end{proof}

Now we complete the analysis of the rounding procedure. We need the following key lemma from \cite{AGSS17} (a generalization of a lemma by Gurvits).

\begin{lemma}
\label{lem:Gurvits}
Let $p(y_1,\ldots,y_m)$ be a homogeneous degree-$n$ real stable polynomial with nonnegative coefficients. For $S \subseteq [m]$, let $c_S$ denote the coefficient of the monomial $y^S$. Then
$$ \sum_{S \in {[m] \choose n}} c_S \geq e^{-n} \inf_{y>0: \forall S \in {[m] \choose n}, y^S \geq 1} p(y_1,\ldots,y_m).$$
\end{lemma}

We are going to apply this lemma to the polynomial $p(y_1,\ldots,y_m) = \prod_{i=1}^{n} \left( \sum_{j=1}^{m} w_{ij} x_{ij} y_j \right)$ which is real-stable (just like in \cite{AGSS17}). We prove the following.

\begin{theorem}
\label{thm:matroid-round}
The outcome of Rounding Procedure 1 has expected value at least $\frac{1}{e} (1-\frac{1}{e})^2 OPT$.
\end{theorem}

\begin{proof}
Denote by $J_i$ the set allocated to player $i$, i.e. $J_i = \{ j \in I_i: p_j = i \}$.
Since $J_i \subseteq I_i \in \cI_i$, $J_i$ is also in $\cI_i$; hence the value obtained by player $i$ is $v_i(J_i) = \sum_{j \in J_i} w_{ij}$. It remains to analyze the expected value of the product $\prod_{i=1}^{n} v_i(J_i)$.
\begin{eqnarray*}
\E\left[ \prod_{i=1}^{n} v_i(J_i) \right] = \E\left[ \prod_{i=1}^{n} \sum_{j \in J_i} w_{ij} \right]
 = \E\left[ \prod_{i=1}^{n} \sum_{j=1}^{m} w_{ij} Y_{ij} Z_{ij} \right]
\end{eqnarray*}
where $Y_{ij}, Z_{ij}$ are the indicator variables of $j \in I_i$ and $p_j = i$ as above.
We continue,
\begin{align*}
\E\left[ \prod_{i=1}^{n} \sum_{j=1}^{m} w_{ij} Y_{ij} Z_{ij} \right]
 = & \sum_{j_1,\ldots,j_n=1}^{m} \E\left[ \prod_{i=1}^{n} w_{i j_i} Y_{i j_i} Z_{i j_i} \right] \\
 \geq & \sum_{ \mbox{\tiny distinct } j_1,\ldots,j_n} \E\left[ \prod_{i=1}^{n} w_{i j_i} Y_{i j_i} Z_{i j_i} \right].
\end{align*}
(The last inequality is in fact an equality, because the same item cannot be allocated to multiple players, but we don't need that here.) We appeal to Lemma~\ref{lem:distinct-YZ}: For every choice of distinct $j_1,\ldots,j_n$, $\E\left[ \prod_{i=1}^{n} Y_{i j_i} Z_{i j_i} \right] \geq (1-1/e)^{2n} \prod_{i=1}^{n} x_{i j_i}$. Therefore,
\begin{align*}
&\sum_{ \mbox{\tiny distinct } j_1,\ldots,j_n} \E\left[ \prod_{i=1}^{n} w_{i j_i} Y_{i j_i} Z_{i j_i} \right] \\
\geq &~(1-1/e)^{2n} \sum_{\mbox{\tiny distinct } j_1,\ldots,j_n}  \prod_{i=1}^{n} w_{i j_i} x_{i j_i}.
\end{align*}
The last expression is exactly $(1-1/e)^{2n}$ times the summation of coefficients for monomials $y^S, S \in {[m] \choose n}$ in the polynomial $p(y_1,\ldots,y_m) = \prod_{i=1}^{n} \left( \sum_{j=1}^{m} w_{ij} x_{ij} y_j \right)$. By Lemma~\ref{lem:Gurvits},
\begin{align*}
&\sum_{ \mbox{\tiny distinct } j_1,\ldots,j_n } \prod_{i=1}^{n} w_{i j_i} x_{i j_i} \\
 \geq & ~e^{-n}  \inf_{y>0: \forall S \in {[m] \choose n}, y^S \geq 1} \prod_{i=1}^{n} \left( \sum_{j=1}^{m} w_{ij} x_{ij} y_j \right) \\ = & ~e^{-n} OPT
\end{align*}
for an optimal fractional solution $\bx$. We conclude that
\begin{eqnarray*}
\E\left[ \prod_{i=1}^{n} v_i(J_i) \right] \geq e^{-n} (1-1/e)^{2n} OPT.
\end{eqnarray*}

\end{proof}

\subsection{Sums of weighted matroid rank functions}

In this section we extend our result to the cone generated by matroid rank functions, or equivalently the sums of weighted matroid rank functions. Consider the rounding method that skips the third bullet point in Rounding 1. 

\paragraph{Rounding Procedure 2 \label{rounding:CR-matroid}:}
Given a fractional solution $\bx \in \RR^{n \times m}$,
\begin{itemize}
    \item For each $(i,j) \in [n] \times [m]$ independently, set $X_{ij} = 1$ with probability $x_{ij}$ and $X_{ij} = 0$ otherwise.
    \item For each $j \in [m]$ independently, apply contention resolution in the uniform rank-1 matroid on $[n]$ to the set $P_j = \{i: X_{ij} = 1\}$, to obtain a singleton $\{p_j\}$ (the player tentatively receiving item $j$).
    \item Allocate item $j$ to player $p_j$.
\end{itemize}

It's easy to see that the expected value of the objective function after this new rounding is at least the expected value after the original rounding. 

\begin{lemma}
\label{lem:independent-rounding}
For any weighted matroid rank functions, the outcome of Rounding Procedure 2 is at least the outcome of Rounding Procedure 1.
\end{lemma}

\begin{proof}
Obviously, the set allocated to each agent $i$ in Rounding Procedure 2 contains the set allocated in Rounding Procedure 1. By monotonicity of the valuations, the value obtained by Rounding Procedure 2 dominates the value obtained by Rounding Procedure 1.
\end{proof}

An interesting point about this new rounding is that it only depends on the fractional solution, in other words, oblivious to the matroids. So we can use this rounding more generally for sums of matroid rank functions. Let the valuation of agent $i$ be the summation of several weighted matroid rank functions, $$v_{i}(\bx_i) = \sum_{k=1}^{s_i} \alpha_k v_{ik}(\bx_i)$$ where $$v_{ik}(\bx_i) = \max\{\sum_{j\in I} \omega_{ijk}: I\in \mathcal{L}_{ik}, \b1_I\le{\bf x}_i\}.$$ where $\mathcal{M}_{ik} = ([m], \mathcal{L}_{ik})$ is a matroid and $\omega_{ijk}\ge0$ for all $i\in[n]$, $j\in[m]$, $k\in[s_i]$. Similarly let $P(\mathcal{M}_{ik})$ be the associated matroid polytope, we have the following relaxation. Note that $\bz_{i*k}$ is the vector with $m$ elements $(z_{ijk})_{j=1}^m$.

\begin{align*}
&\max_{\bx, \bz}  \inf_{\by \in \RR_+^m: y^S\ge1,\forall S\in\binom{[m]}{n}} \prod_{i=1}^n&\left(\sum_{k=1}^{s_i}\sum_{j=1}^{m} w_{ijk} z_{ijk} y_j \right),\\
& s.t.~~~~~ \bz_{i*k} \in P(\cM_{ik}) & \forall i \in [n] , k\in [s_i]\\
& ~~~~~~~~~ \sum_{i=1}^{n} x_{ij}\le 1 & \forall j \in [m] \\
& ~~~~~~~~~ z_{ijk}\le x_{ij} & \forall i \in [n], \forall j \in [m], \forall k\in[s_i] \\
\end{align*}

It is easy to prove this is a relaxation of the Nash welfare maximization problem.

\begin{lemma}
\label{lem:relaxation-sum}
The optimal solution of the above program is at least the optimal solution of the Nash Social Welfare maximization problem.
\end{lemma}

\begin{proof}
Suppose that $\bx \in \{0,1\}^{n \times m}$ is the optimal allocation. Let $\bz_{i*k} \in \{0,1\}^{m}$ be the indicator of a maximal independent set in matroid $\cM_{ik}$ such that $\bz_{i*k} \leq \bx_i$.
Then we have $\bz_{i*k} \in P(\cM_{ik})$, $\sum_{i=1}^{n} x_{ij} \leq 1$, $z_{ijk}\le x_{ij}$. Also, for every $\by \in \RR_+^m$ such that $y^S \geq 1$  for all $S \in {[m] \choose n}$, we have
\begin{align*}
& \prod_{i=1}^n\left(\sum_{k=1}^{s_i}\sum_{j=1}^{m} w_{ij} z_{ijk} y_j\right) \\
= & \sum_{\substack{k_1, \ldots, k_n \\ 1\le k_i\le s_i}}\sum_{j_{1},\ldots,j_n=1}^{m} \prod_{i=1}^{n} (w_{ij_ik_i} z_{ij_ik_i} y_{j_i}) \\
 = & \sum_{\substack{k_1, \ldots, k_n \\ 1\le k_i\le s_i}}\sum_{\substack{\mbox{\tiny distinct } \\ j_1,\ldots,j_n}} y^{\{j_1,\ldots,j_n\}} \prod_{i=1}^{n} (w_{ij_ik_i} z_{ij_ik_i}) \\  
\geq & \sum_{\substack{k_1, \ldots, k_n \\ 1\le k_i\le s_i}}\sum_{\substack{\mbox{\tiny distinct } \\ j_1,\ldots,j_n}} \prod_{i=1}^{n} (w_{ij_ik_i} z_{ij_ik_i})   \\  
= & \prod_{i=1}^n\left(\sum_{k=1}^{s_i}\sum_{j=1}^{m} w_{ijk_i} z_{ijk}\right) 
\end{align*}
where the inequality holds because of the constraint $y^S \geq 1$, and the equality between summations over all $n-$tuples and distinct $n$-tuples holds because we cannot have $z_{ijk} = z_{i' j'k'} = 1$ for $i \neq i'$ and $j = j'$ (each item is assigned to only one player).
Therefore, the optimum of the above program is at least the integer optimum.
\end{proof}

Moreover, similarly, by Lemma~\ref{lem:concave-convex} and standard techniques using separation oracles for matroid polytopes, we can solve this program efficiently (to arbitrary precision). Finally, we use the Rounding Procedure 2 to round a fractional solution of the program.

\begin{theorem}
The outcome of the Rounding 2 has expected value at least $\frac1e(1-\frac1e)^2OPT$.
\end{theorem}

\begin{proof}
Denote by $J_i$ the set allocated to player $i$. Fix $k_1, \ldots, k_n$ such that $1\le k_i\le s_i$, consider the expected value of $\E\left[\prod_{i=1}^n v_{ik_i}(J_i)\right]$ after the rounding. Since $x_{ij} \ge z_{ijk}$, we know that the outcome of assigning each item with probability $x_{ij}$ is at least the outcome of assigning each item with probability $z_{ijk_i}$. Moreover, by Lemma~\ref{lem:independent-rounding}, we know that the outcome can be further lower-bounded by the outcome of Rounding Procedure 1, which is at least 
$$(1-1/e)^{2n}\sum_{\substack{j_1,\ldots,j_n \\\mbox{ \tiny distinct}}} \prod_{i=1}^{n} w_{i j_i k_i} z_{i j_i k_i}.$$
  by Theorem~\ref{thm:matroid-round}. Putting all the possible choices of $k_1,\ldots,k_n$ together, we have 
  
  \begin{align*}
       &\E\left[\prod_{i=1}^n\sum_{k_i=1}^{s_i}v_{ik_i}(J_i)\right] \\ =& \sum_{\substack{k_1, \ldots, k_n \\ 1\le k_i\le s_i}}\E\left[\prod_{i=1}^n v_{ik_i}(J_i)\right]  \\
       \ge& \sum_{\substack{k_1, \ldots, k_n \\ 1\le k_i\le s_i}}(1-1/e)^{2n}\sum_{\substack{j_1,\ldots,j_n \\\mbox{ \tiny distinct}}} \prod_{i=1}^{n} w_{i j_i k_i} z_{i j_i k_i} \\
        =& (1-1/e)^{2n}\sum_{\substack{j_1,\ldots,j_n \\\mbox{ \tiny distinct}}} \prod_{i=1}^{n} \sum_{k=1}^{s_i}w_{i j_i k} z_{i j_i k} \\
        \ge& e^{-n}(1-1/e)^{2n}\inf_{\substack{\by \in \RR_+^m: y^S\ge1, \\ \forall S\in\binom{[m]}{n}}} \prod_{i=1}^n\left(\sum_{k=1}^{s_i}\sum_{j=1}^{m} w_{ijk} z_{ijk} y_j \right)
  \end{align*}
 where the last inequality is due to Lemma~\ref{lem:Gurvits}.
 Finally, due to Lemma~\ref{lem:relaxation-sum}, the last quantity is at least $\frac{1}{e} (1-\frac{1}{e})^2 OPT$.
\end{proof}

\section{Bipartite matching with a matroid constraint}

In this section we discuss the case of valuation functions defined by a bipartite matching problem with a matroid constraint. More specifically, consider $n$ bipartite graphs $G_i = (L_i, R_i, E_i)$, $1 \le i\le n$, $L_i = [m]$, with a non-negative weight $w_{i,j,k}$ on each edge $(j, k)\in E_i$, and a matroid $\cM_i = (R_i, \cI_i)$ defined on on the right-hand side $R_i$ of each graph. Denote by $v_i$ the valuation function for player $i$. For any $S\subseteq L_i = [m]$, we define $v_i(S)$ to be the maximum weight of a matching $M_i$ such that vertices matched on the left are a subset of $S$, and the vertices matched on the right are an independent set in $\cM_i$.

Denote by $P(\cM_i)$ the associated matroid polytope. Let $\bz \in \mathbb{R}^{E_1} \times \ldots \times \mathbb{R}^{E_n}$ be a vector indexed by edges of the bipartite graphs $G_i$ where $z_{i, j, k}$ represents whether edge $(j, k) \in E_i$ is in the matching $M_i$ or not. Consider the following mathematical program, 

\begin{align*}
&\max_{\bz \in \mathbb{R}^{E_1} \times \ldots \times \mathbb{R}^{E_n}}\inf_{\substack{\by>0:y^S\ge1\\ \forall S\in\binom{[m]}{n}}}\prod_{i=1}^n&\left(\sum_{(j, k)\in E_i}w_{i,j,k}z_{i, j, k}y_j\right) \\
& s.t. ~~~ \sum_{i,k: (j, k)\in E_i} z_{i, j, k}\le1 & \forall j\in [m] \\
& ~~~~~~~ u_{i,k} = \sum_{j: (j,k)\in E_i}z_{i, j, k}\le1 & \forall 1\le i\le n, k\in R_i \\
& ~~~~~~~~~~~ {\mathbf u}_{i} \in P(\cM_i) & \forall 1\le i\le n \\
& ~~~~~~~~ 0\le z_{i, j, k}\le 1 & \forall 1\le i\le n, (j,k)\in E_i\\
\end{align*}

First we show that the program above is indeed a relaxation of Nash Social Welfare maximization.

\begin{lemma}
\label{lem:relaxation-matching}
The optimal solution of the above program is at least the optimal solution of the Nash Social Welfare maximization problem.
\end{lemma}

\begin{proof}
Suppose that $\bx \in \{0,1\}^{n \times m}$ is the optimal allocation. Assume one of the maximal matchings in graph $G_i$ subject to the matroid constraint $\cM_i$, corresponding to this assignment, is $\bz_{i} \in \{0,1\}^{E_i}$.
Then it is easy to verify that $\bz_i$ satisfies all the constraints in the above program. Also, for every $\by \in \RR_+^m$ such that $y^S \geq 1$  for all $S \in {[m] \choose n}$, we have
\begin{align*}
    &\prod_{i=1}^{n} \left( \sum_{(j, k)\in E_i} w_{i,j,k} x_{i,j,k} y_j \right) \\ = & 
 \sum_{\substack{(j_i, k_i)\in E_i\\ 1\le i\le n}} \prod_{i=1}^{n} (w_{i,j_i,k_i} x_{i,j_i,k_i} y_{j_i}) \\
 = &\sum_{\substack{\mbox{\tiny distinct } j_1,\ldots,j_n \\ (j_i, k_i)\in E_i}} y^{\{j_1,\ldots,j_n\}} \prod_{i=1}^{n} (w_{i,j_i,k_i} x_{i,j_i,k_i})  \\
\geq &\sum_{\substack{\mbox{\tiny distinct } j_1,\ldots,j_n \\ (j_i, k_i)\in E_i}} \prod_{i=1}^{n} (w_{i,j_i,k_i} x_{i,j_i,k_i}) \\
= & \prod_{i=1}^{n} \left( \sum_{(j, k)\in E_i} w_{ijk} x_{ijk}\right)
\end{align*}   
where the inequality holds because of the constraint $y^S \geq 1$, and the equality between summations over all $n-$tuples and distinct $n$-tuples holds because we cannot have $z_{i,j,k} = z_{i', j',k'} = 1$ for $i \neq i'$ and $j = j'$ (we assign an item to only one player).
Therefore, the optimum of the above program is at least the integer optimum.
\end{proof}

Similarly, by Lemma~\ref{lem:concave-convex} and standard techniques like \cite{AGSS17} using separation oracles for matroid polytopes we can solve this program efficiently.

Next we are going to use contention resolution to design a rounding procedure. Notice that the constraints for ${\bf z}$  correspond to the polytope of the intersection of two matroids $P_{\cM^{(L)}\cap \cM^{(R)}}$: one partition matroid $\mathcal{M}^{(L)}$ encoding that at most $1$ edge can be chosen among all the edges incident to the same item $j$ (across all the graphs $G_i$); and another (general) matroid $\mathcal{M}^{(R)}$ such that among the edges in each $G_i$, at most one edge incident to each vertex on the right can be chosen, and the matched vertices form an independent set in $\mathcal{M}_i$. ($\mathcal{M}^{(R)}$ can be constructed by taking a disjoint union of the matroids $\mathcal{M}_i$, and adding parallel copies of elements representing the different edges incident to a given vertex on the right.) 

We can further define the allocation constraint on item $j$ as $\cM^{(L)}_j$ and the matroid constraint on the right-hand side of graph $G_i$ as  $\cM^{(R)}_i$. Our rounding algorithm is as follows. 

\

\paragraph{Rounding Procedure 3 \label{rounding:CR-matroid}: Rounding for bipartite matching with matroid constraints valuations.} Given a fractional solution $\bz$,

\begin{itemize}
\item For each edge $(j, k)\in E_i$, $1\le i\le n$ independently, set $X_{ijk} = 1$ with probability $z_{i,j,k}$ and $X_{i,j,k} = 0$ otherwise.
\item For each item $j \in[m]$ independently, run a $(1, 1-e^{-1})$-balanced contention resolution algorithm on the set $P_{j} = \{(i',j',k'): j' = j, X_{i',j',k'} = 1\}$ and with the matroid $\mathcal{M}^{(L)}_j$. Denote by $Y_{i,j,k}$ the indicator random variable of the event that edge $(i, j, k)$ is selected in this contention resolution.
\item For each agent $i\in[n]$ independently, run a $(1, 1-e^{-1})$-balanced contention resolution algorithm on the set $Q_i = \{(i',j',k'): i'=i, X_{i',j',k'} = 1\}$ and with the matroid $\mathcal{M}^{(R)}_i$. Denote by $Z_{i,j,k}$ the indicator random variables of the event that edge $(i, j, k)$ is selected in this contention resolution.
\item Allocate item $j$ to agent $i$, if $Y_{i,j,k} Z_{i,j,k} = 1$ for some $k \in R_i$.
\end{itemize}

We can prove that this rounding algorithm is a good approximation using a similar strategy as Theorem~\ref{thm:matroid-round}.

\begin{theorem}
The outcome of the rounding above has expected value at least $\frac1e(1-\frac1e)^2OPT$.
\end{theorem}

\begin{proof}

First we will prove a similar inequality as Lemma~\ref{lem:distinct-YZ}: consider $n$ distinct items $j_1, \ldots, j_n$ where $j_i$ is allocated to player $i$. That means there must be some edge $(j_i, k_i)$ chosen in the rounding. We will show the following inequality:  

$$\mathbb{E}\left[\prod_{i=1}^n Y_{i, j_i, k_i}Z_{i, j_i, k_i}\right]\ge (1-\frac1e)^{2n}\prod_{i=1}^nz_{i,j_i,k_i}$$

As we have discussed in the rounding procedure, for any edge $(j, k)\in E_i$, $Y_{i, j, k}$ is a function of variables in $P_{j}$, and $Z_{i, j, k}$  is a function of variables in $Q_i$. Furthermore, by the monotonicity of the contention resolution scheme, $$\mathbb{E}[Y_{i, j, k} \mid X_{i',j,k'}: (i',j,k') \in P_{j} ]$$ and $$\mathbb{E}[Z_{i, j, k} \mid X_{i,j',k'}: (i, j',k') \in Q_i ]$$ are both functions that are non-increasing in the variables $X_{i',j',k'}$, except $X_{i,j,k}$ itself which must be $1$, otherwise $Y_{i,j,k} = Z_{i,j,k} = 0$. 

Let $\cX$ denote the event that $\forall i\in [n]$, $X_{i, j_i, k_i} = 1$. As in the previous section, we have a product probability space conditioned on $\cal X$, and non-increasing functions on this product space. By repeated applications of the FKG inequality, we have 

\begin{align*}
\mathbb{E}\left[\prod_{i=1}^n Y_{i, j_i, k_i}Z_{i, j_i, k_i}|\cX\right] &\ge \prod_{i=1}^n \mathbb{E}\left[Y_{i, j_i, k_i}|\cX\right]\mathbb{E}\left[Z_{i, j_i, k_i}|\cX\right].
\end{align*}

By the properties of contention resolution scheme we use, for any $1\le i\le n$,  $$\mathbb{E}\left[Y_{i, j_i, k_i}|\cX\right] = \mathbb{E}\left[Y_{i, j_i, k_i}|X_{i, j_i, k_i}\right] \ge 1-1/e,$$ $$\mathbb{E}\left[Z_{i, j_i, k_i}|\cX\right] = \mathbb{E}\left[Z_{i, j_i, k_i}|X_{i, j_i, k_i}\right] \ge 1-1/e.$$ 
Putting this together we have 

\begin{align*}
    &\mathbb{E}\left[\prod_{i=1}^n Y_{i, j_i, k_i}Z_{i, j_i, k_i}\right] \\ = & ~  \mathbb{E}\left[\prod_{i=1}^n Y_{i, j_i, k_i}Z_{i, j_i, k_i}|\cX\right]\Pr[\cX] \\ \ge & ~ (1-\frac1e)^{2n}\Pr[\cX] = (1-\frac1e)^{2n}\prod_{i=1}^nz_{i,j_i,k_i}.
\end{align*}

Now we can bound the expect product of valuation over all agents, after the rounding, which is 
\begin{align*}
    &\sum_{\substack{\mbox{\tiny distinct } j_1, \ldots, j_n \\  (j_i, k_i)\in E_i}} \mathbb{E}\left[\prod_{i=1}^n Y_{i, j_i, k_i}Z_{i, j_i, k_i}\right] \\ \ge & ~(1-\frac1e)^{2n}\sum_{\substack{\mbox{\tiny distinct } j_1, \ldots, j_n \\  (j_i, k_i)\in E_i}}\prod_{i=1}^nz_{i, j, k}.
\end{align*} 

By Lemma~\ref{lem:Gurvits} (Generalized Gurvits' Lemma), we know that is at least 
$$(1-\frac1e)^{2n}e^{-n}\inf_{\by>0:y^S\ge1,\forall S\in\binom{[m]}{n}}\prod_{i=1}^n\left(\sum_{(j, k)\in E_i}w_{i.j,k}z_{i, j, k}y_j\right)$$ which is a $(1-1/e)^{2n}e^{-n}$ approximation for the objective function. 
\end{proof}

Moreover, similar to the case of weighted matroid rank functions, we can prove that skipping the third bullet point in the rounding will not decrease the outcome and the procedure will not depend on the valuation functions any more. So similarly, we can still use the new rounding method for sums of bipartite matching valuations with a matroid constraint. The proof is quite similar to the matroid case and we omit the proof here.

\bibliographystyle{plain}
\bibliography{main} 

\begin{thebibliography}{10}

\bibitem{AMGV18}
Nima Anari, Tung Mai, Shayan Oveis~Gharan, and Vijay~V. Vazirani.
\newblock Nash social welfare for indivisible items under separable,
  piecewise-linear concave utilities.
\newblock In {\em Proceedings of the 29th Annual {ACM}-{SIAM} Symposium on
  Discrete Algorithms}, pages 2274--2290. {ACM}, January 2018.

\bibitem{AGSS17}
Nima Anari, Shayan Oveis~Gharan, Amin Saberi, and Mohit Singh.
\newblock Nash social welfare, matrix permanent, and stable polynomials.
\newblock In {\em 8th Innovations in Theoretical Computer Science Conference
  (ITCS 2017)}. Schloss Dagstuhl-Leibniz-Zentrum fuer Informatik, 2017.

\bibitem{barman2018finding}
Siddharth Barman, Sanath~Kumar Krishnamurthy, and Rohit Vaish.
\newblock Finding fair and efficient allocations.
\newblock In {\em Proceedings of the 2018 ACM Conference on Economics and
  Computation}, pages 557--574, 2018.

\bibitem{CCG18}
Bhaskar~Ray Chaudhury, Yun~Kuen Cheung, Jugal Garg, Naveen Garg, Martin Hoefer,
  and Kurt Mehlhorn.
\newblock On fair division for indivisible items.
\newblock In {\em 38th IARCS Annual Conference on Foundations of Software
  Technology and Theoretical Computer Science (FSTTCS 2018)}. Schloss
  Dagstuhl-Leibniz-Zentrum fuer Informatik, 2018.

\bibitem{CVZ14}
Chandra Chekuri, Jan Vondr{\'a}k, and Rico Zenklusen.
\newblock Submodular function maximization via the multilinear relaxation and
  contention resolution schemes.
\newblock {\em SIAM Journal on Computing}, 43(6):1831--1879, 2014.

\bibitem{cole2017convex}
Richard Cole, Nikhil Devanur, Vasilis Gkatzelis, Kamal Jain, Tung Mai, Vijay~V
  Vazirani, and Sadra Yazdanbod.
\newblock Convex program duality, fisher markets, and nash social welfare.
\newblock In {\em Proceedings of the 2017 ACM Conference on Economics and
  Computation}, pages 459--460, 2017.

\bibitem{CG18}
Richard Cole and Vasilis Gkatzelis.
\newblock Approximating the nash social welfare with indivisible items.
\newblock {\em SIAM Journal on Computing}, 47(3):1211--1236, 2018.

\bibitem{Feige08}
Uriel Feige.
\newblock On allocations that maximize fairness.
\newblock In {\em SODA}, volume~8, pages 287--293. Citeseer, 2008.

\bibitem{FeigeV10}
Uriel Feige and Jan Vondr{\'{a}}k.
\newblock The submodular welfare problem with demand queries.
\newblock {\em Theory Comput.}, 6(1):247--290, 2010.

\bibitem{FNW78}
M.~L. Fisher, G.~L. Nemhauser, and L.~A. Wolsey.
\newblock An analysis of approximations for maximizing submodular set functions
  -- {II}.
\newblock {\em Mathematical Programming Study}, 8:73--87, 1978.

\bibitem{GHM19}
Jugal Garg, Martin Hoefer, and Kurt Mehlhorn.
\newblock Approximating the nash social welfare with budget-additive
  valuations.
\newblock In {\em Proceedings of the Twenty-Ninth Annual ACM-SIAM Symposium on
  Discrete Algorithms}, pages 2326--2340. SIAM, 2018.

\bibitem{GKK20}
Jugal Garg, Pooja Kulkarni, and Rucha Kulkarni.
\newblock Approximating nash social welfare under submodular valuations through
  (un) matchings.
\newblock In {\em Proceedings of the fourteenth annual ACM-SIAM symposium on
  discrete algorithms}, pages 2673--2687. SIAM, 2020.

\bibitem{gurvits2005complexity}
Leonid Gurvits.
\newblock On the complexity of mixed discriminants and related problems.
\newblock In {\em International Symposium on Mathematical Foundations of
  Computer Science}, pages 447--458. Springer, 2005.

\bibitem{HSS11}
Bernhard Haeupler, Barna Saha, and Aravind Srinivasan.
\newblock New constructive aspects of the lov{\'{a}}sz local lemma.
\newblock {\em J. {ACM}}, 58(6):28:1--28:28, 2011.

\bibitem{Kel97}
Frank Kelly.
\newblock Charging and rate control for elastic traffic.
\newblock {\em Eur. Trans. Telecommun.}, 8(1):33--37, 1997.

\bibitem{Lee17}
Euiwoong Lee.
\newblock Apx-hardness of maximizing nash social welfare with indivisible
  items.
\newblock {\em Information Processing Letters}, 122:17--20, 2017.

\bibitem{Leme18}
Renato~Paes Leme.
\newblock Personal communication, 2020.

\bibitem{Murota15}
Kazuo Murota.
\newblock Problems for ``{D}iscrete {C}onvex {A}nalysis”, {HIM} summer
  school, 2015.

\bibitem{Nash50}
J.~Nash.
\newblock The bargaining problem.
\newblock {\em Econometrica}, 18(2):155--162, April 1950.

\bibitem{Var74}
H.R. Varian.
\newblock Equality, envy, efficiency.
\newblock {\em Journal of Economic Theory}, 9(1):63--91, 1974.

\bibitem{Von08}
Jan Vondr\'ak.
\newblock Optimal approximation for the submodular welfare problem in the value
  oracle model.
\newblock In {\em Proceedings of the Annual ACM Symposium on Theory of
  Computing}, pages 67--74, 2008.

\end{thebibliography}

\appendix

\section{High Dimensional matching with matroid constraints} \label{appen:high-dimension}

In this section we consider the most general class of valuation functions we can handle -- high dimensional matching with matroid constraints.

Let us represent the valuation function of each player by a weighted k-partite hypergraph and $k-1$ matroids. For player $i$, consider a hypergraph $G_i = (V_i, E_i)$ where $V_i = S_i \cup T_{i,1} \cup \ldots \cup T_{i, k-1}$, $S_i = [m]$ is the set of items, and $T_{i, 1}, \ldots, T_{i,k-1}$ are disjoint sets; each hyperedge in $E_i$ contains exactly 1 vertex from each of $S_i, T_{i,1}, \ldots, T_{i,k-1}$. We have a non-negative weight $w_{e}^{(i)}$ on the hyperedge $e \in E_i$, and a matroid constraint $\cM_{i, \ell} = (T_{i, \ell}, \cI_{i, \ell})$ defined on each $T_{i, \ell}$ for $1 \le \ell \le k$. Denote by $v_i$ the valuation function for agent $i$. For any $S \subseteq S_i = [m]$, we define $v_i(S)$ to be the maximum weight of a matching of hyperedges in $G_i$ such that vertices matched in $S_i$ are a subset of $S$, and the vertices matched in $T_{i, \ell}$ are an independent set in $\cM_{i, \ell}$ for every $1 \le \ell \le k-1$.

Denote by $P(\cM_{i, \ell})$ the matroid polytope for matroid $\cM_{i, \ell}$. Let $\bz \in \mathbb{R}^{E_1} \times ... \times \mathbb{R}^{E_n}$ be an allocation vector where $z_{e}^{(i)} = 1$ indicates that edge $e\in E_i$ is in the matching representing the value of $v_i(S)$. Consider the following mathematical program:

\begin{align*}
&\max_{\bz \in \mathbb{R}^{E_1} \times...\times \mathbb{R}^{E_n}}  \inf_{\substack{\by>0:y^S\ge1\\ \forall S\in\binom{[m]}{n}}}\prod_{i=1}^n\left(\sum_{\substack{e \in E_i \\ e_i \text{ incident on } j \text{ in } S_i}}w_{e}^{(i)}z_{e}^{(i)}y_{j}\right) \\
& s.t.  \sum_{1\le i\le n, e\in E_i}z_{e}^{(i)}\le1, ~~~~~~~~~~~~~~~~~~~~~~~~~~~~~~~\forall j\in [m] \\
& ~~~~~ u^{i,\ell}_v = \sum_{e \in E_i: v \in e} z_{e}^{(i)}\le1, \forall 1\le i\le n, 1\le \ell \le k, v\in T_{i,\ell} \\
& ~~~~~ {\mathbf u}^{i,\ell} \in P(\cM_{i, \ell}),~~~ \forall 1\le i\le n, 1\le \ell \le k \\
& ~~~~~ 0\le z_{e}^{(i)}\le 1, ~~~~~~~~~~~\forall 1\le i\le n, e\in E_i\\
\end{align*}

First we prove that the program above is indeed a relaxation of Nash social welfare maximization.

\begin{lemma}
\label{lem:relaxation-matching}
The optimal solution of the above program is at least the optimal solution of the Nash Social Welfare problem.
\end{lemma}

\begin{proof}
Suppose that $\bx \in \{0,1\}^{n \times m}$ is the optimal allocation. Let $\bz_{i} \in \{0,1\}^{|E_i|}$ denote a maximum-weight matching in $G_i$ attaining the value of agent $i$ and satisfying the respective matroid constraints.
It is easy to verify that $\bz_i$ satisfies all the constraints in the above program. Also, for every $\by \in \RR_+^m$ such that $y^S \geq 1$  for all $S \in {[m] \choose n}$, we have
\begin{align*} 
&\prod_{i=1}^n\left(\sum_{j \in S_i} \sum_{\substack{e \in E_i: e \cap S_i = \{j\}}}w_{e}^{(i)}z_{e}^{(i)}y_{j}\right)\\
= & \sum_{j_1 \in S_1, \ldots, j_n \in S_n} \sum_{\substack{e_i \in E_i: e_i \cap S_i = \{j_i\}}} \prod_{i=1}^{n} \left(w_{e_i}^{(i)}z_{e_i}^{(i)}y_{j_i}\right) \\
= & \sum_{\substack{\mbox{\tiny distinct } j_1, \ldots, j_n \\  e_i\in E_i: e_i \cap S_i = \{j_i\}}} y^{\{j_1,\ldots,j_n\}} \prod_{i=1}^{n} \left(w_{e_i}^{(i)}z_{e_i}^{(i)}\right) \\
\ge & \sum_{\substack{\mbox{\tiny distinct } j_1, \ldots, j_n \\  e_i\in E_i: e_i \cap S_i = \{j_i\}}} \prod_{i=1}^{n} \left(w_{e_i}^{(i)}z_{e_i}^{(i)}\right)
= \prod_{i=1}^n\left(\sum_{e \in E_i}w_{e}^{(i)}z_{e}^{(i)}\right)  
\end{align*}

where the inequality holds because of the constraint $y^S \geq 1$, and the equality between summations over all $n-$tuples and distinct $n$-tuples holds because we cannot have $z_{e_i}^{(i)} = z_{e_{i'}}^{(i')} = 1$ for $i \neq i'$, $e_i$ and $e_{i'}$ incident on the same item $j$.
Therefore, the optimum of the above program is at least the integer optimum.
\end{proof}

Similarly, by Lemma~\ref{lem:concave-convex} and standard convex optimization techniques, using separation oracles for matroid polytopes, we can solve this program efficiently.

Next we are going to use contention resolution scheme to design a rounding procedure. Notice that the constraints for variable ${\bf z}$ actually corresponds to the intersection of $k$ matroid polytopes $\cap_{i=0}^{k-1} P(\cM^{(i)})$: one partition matroid $\mathcal{M}^{(0)}$ for the assignment constraint (vertex sets $V_{1,0}, \ldots, V_{n,0}$) saying that each item is allocated to at most one agent and that an agent can allocate at most one edge to this item, a general matroid constraint $\mathcal{M}^{(\ell)}$ on $V_{1,\ell}, \ldots, V_{n,\ell}$ for any $1\le \ell \le k-1$ saying that each vertex is incident to at most one matching hyperedge and for each agent $i$, the matched subset of each $T_{i,\ell}$ is independent in $\cM_{i,\ell}$.

Notice that each of those matroids is a union of $n$ matroids, one for each agent. The partition matroid $\cM^{(0)}$ can be split it up as $\cM^{(0)} = \bigcup_{j=1}^{m} \cM^{(0)}_{j}$ where $\cM^{(0)}_{j}$ expresses the allocation constraint for item $j$. For the other matroids $\cM^{(\ell)}$ , we can write $\cM^{(\ell)} = \cup_{i=1}^{n} \cM^{(\ell)}_i$ where $\cM^{(\ell)}_i$ expresses the $\ell$-th matroid constraint for agent $i$. Our rounding algorithm works as follows. 

\

\paragraph{Rounding Procedure 4 \label{rounding:CR-matroid}: Rounding for k-partite matching with matroid constraints.} Given a fractional solution $\bz$,

\begin{itemize}
\item For every agent $1\le i\le n$ and each edge $e \in E_i$,  independently set $X_{e}^{(i)} = 1$ with probability $z_{e}^{(i)}$ and $X_{e}^{(i)} = 0$ otherwise.
\item For each item $j\in[m]$, independently run a $(1, 1-1/e)$ contention resolution algorithm on the set $P_j = \{(e, i): e\in E_i, 1 \leq i \leq n, X_{e}^{(i)} = 1, j \in e \}$ with the matroid constraint $\mathcal{M}^{(0)}_j$. Denote by $Y_{e}^{(i)}$ the indicator random variable for the event that edge $e \in E_i$ is selected in this procedure.
\item For each agent $i \in [n]$ and $1 \le \ell \le k-1$, independently run a $(b, \frac{1}{b}(1-e^{-b}))$ contention resolution algorithm for $b = \frac{1}{k-1}$ on the set $Q_i = \{e\in E_i: X_{e}^{(i)} = 1\}$ with the matroid constraint $\mathcal{M}^{(i)}_\ell$. Denote by $Z_{e, \ell}^{(i)} = 1$ the indicator random variable for the event that $e \in E_i$ is selected in this procedure.
\item We allocate item $j$ to agent $i$, if there is edge $e \in E_i$ incident to this item $j$ such that $Y_{e}^{(i)} Z_{e}^{(i)}  = 1$ where $Z_{e}^{(i)} = \prod_{\ell=1}^{k} Z_{e, \ell}^{(i)}$.
\end{itemize}

We can prove that this rounding algorithm is a good approximation using a similar strategy as Theorem~\ref{thm:matroid-round}. 

\begin{theorem}
The outcome of Rounding Procedure 4 has expected value at least $(e^2 k)^{-n} OPT$.
\end{theorem}

\begin{proof}

First we will prove a similar inequality as Lemma~\ref{lem:distinct-YZ}: consider $n$ distinct items $j_1, \ldots, j_n$ where $j_i$ is allocated to player $i$. That means there must be some edge $e_i$ such that it is chosen in the rounding and $e_i$ is incident on $j_i$ in $S_i$. We will show the following inequality:  

$$\mathbb{E}\left[\prod_{i=1}^n Y_{e_i}^{(i)}Z_{e_i}^{(i)}\right]\ge (ek)^{-n}\prod_{i=1}^nz_{e_i}^{(i)}$$



As we have discussed in the rounding procedure, for any edge $e\in E_i$ incident to item $j$, $Y_{e}^{(i)}$ is a function of variables in $P_{j}$, and $Z_{i, j, k}$  is a function of variables in $Q_i$. Furthermore, by the monotonicity of the contention resolution scheme, $$\mathbb{E}[Y_{e}^{(i)} \mid X_{e'}^{(i')}: (e',i') \in P_{j} ]$$ and $$\mathbb{E}[Z_{e}^{(i)} \mid X_{e'}^{(i')}: e'\in Q_i ]$$ are both functions that are decreasing in the variables $X_{e'}^{(i')}$, except $X_{e}^{(i)}$ itself which must be $1$ otherwise $Y_{e}^{(i)} = Z_{e}^{(i)} = 0$. 

Let $\cX$ denote the event that $\forall i\in [n]$, $X_{e_i}^{(i)} = 1$. By the FKG inequality applied repeatedly, we get

\begin{align*}
\mathbb{E}\left[\prod_{i=1}^n Y_{e_i}^{(i)}Z_{e_i}^{(i)}|\cX\right] \ge \prod_{i=1}^n \mathbb{E}\left[Y_{e_i}^{(i)}|\cX\right]\mathbb{E}\left[Z_{e_i}^{(i)}|\cX\right].
\end{align*}

By the properties of the contention resolution schemes we use, for any $1\le i\le n$,  $$\mathbb{E}\left[Y_{e_i}^{(i)}|\cX\right] = \mathbb{E}\left[Y_{e_i}^{(i)}|X_{e}^{(e)}\right] \ge 1-1/e$$ and $$\mathbb{E}\left[Z_{e_i}^{(i)}|\cX\right]\ge \frac1{k-1}\left(\frac{1-\exp\{-\frac1{k-1}\}}{1/(k-1)}\right)^{k-1} $$ 
$$ \geq \frac{1}{k-1} \left( \frac{1 - 1 / (1 + \frac{1}{k-1})}{1 / (k-1) } \right)^{k-1} $$
$$ =  \frac{1}{k-1} \left( \frac{k-1}{k} \right)^{k-1} \geq \frac{1}{ek}.$$

Putting this together we have 

\begin{align*}\mathbb{E}\left[\prod_{i=1}^n Y_{e_i}^{(i)}Z_{e_i}^{(i)}\right] &= \mathbb{E}\left[\prod_{i=1}^n Y_{e_i}^{(i)}Z_{e_i}^{(i)}|\cX\right]\Pr[\cX] \\ &\ge (ek)^{-n}\Pr[\cX] \geq (e^2 k)^{-n}\prod_{i=1}^nz_{e_i}^{(i)}.
\end{align*}

Now we can bound the product of valuations over all agents, after rounding, which is 
\begin{align*}
&\sum_{\substack{\mbox{\tiny distinct } j_1, \ldots, j_n \\  e_i\in E_i: e_i \cap S_i = \{j_{i}\}}} \mathbb{E}\left[\prod_{i=1}^n Y_{e_i}^{(i)}Z_{e_i}^{(i)}\right] \\ \ge & ~ (ek)^{-n}\sum_{\substack{\mbox{\tiny distinct } j_1, \ldots, j_n \\  e_i\in E_i: e_i \cap S_i = \{j_{i}\}}}\prod_{i=1}^nz_{e_i}^{(i)}. 
\end{align*}

By Lemma~\ref{lem:Gurvits} (Generalized Gurvits' Lemma), we know that is at least 
$$(ek)^{-n}e^{-n}\inf_{\substack{\by>0:y^S\ge1, \\ \forall S\in\binom{[m]}{n}}}\prod_{i=1}^n\left(\sum_{\substack{e \in E_i \\ e_i \text{ incident on } j \text{ in } S_i}}w_{e}^{(i)}z_{e}^{(i)}y_{j}\right)$$ which is a $(e^2 k)^{-n}$ approximation for objective function of the program. 
\end{proof}

Finally, similar to the case of weighted matroid rank functions, we can show that skipping the third bullet point in the rounding will not decrease the outcome and the procedure will not depend on the valuation functions any more. So we can still use this new rounding to extend the result to the cone of $k$-dimensional matching valuations with matroid constraints. The proof is quite similar to the matroid case, so we omit it here.

\section{Open counting problems and their connection to Nash Social Welfare}

Let us discuss here some of the directions that our work leaves open. 
One way to resolve the issue of finding a good allocation of items might be to apply the method of conditional expectations.
Given a fractional solution of our concave-convex program, suppose that we apply the following rounding procedure (just like in \cite{AGSS17}).

\

\paragraph{Rounding Procedure 0 \label{rounding:CR-matroid}: simple rounding.}
Given a fractional solution $\bx \in \RR^{n \times m}$:\\
For each $j \in [m]$ independently, allocate item $j$ to at most one agent, agent $i$ with probability $x_{ij}$.

\

\begin{lemma}
For arbitrary monotone valuations, Rounding Procedure 0 provides expected Nash Social Welfare at least as large as Rounding Procedure 2.
\end{lemma}

\begin{proof}
We show that the two procedures can be coupled in a way that every agent in Rounding Procedure 2 receives a subset of what they  receive in Rounding Procedure 0.

Consider Rounding Procedure 2: we first set each variable $X_{ij}$ independently to be $1$ with probability $x_{ij}$, and $0$ otherwise. Then for each fixed $j$, we use contention resolution to select one of the variables $X_{ij}$ equal to $1$ (if any) and allocate item $j$ to this agent $i$. An explicit description of such a procedure can be found in \cite{FeigeV10}, Section 1.2: Given initial probabilities $x_{1j}, \ldots, x_{nj}$, agent $i$ receives item $j$ with probability $x'_{ij} = \frac{x_{ij}}{\sum_i x_{ij}} (1 - \prod_i (1-x_{ij})) \geq (1-1/e) x_{ij}$.

Now let us consider the event where no agent receives item $j$: This occurs with probability $1 - \sum_i x'_{ij} = \prod_i (1-x_{ij})$. We modify the scheme by allocating item $j$ in case the contention resolution scheme yielded no allocation, to agent $i$ with additional probability $x_{ij} - x'_{ij}$; with the remaining probability $\prod_i (1-x_{ij}) - \sum_i (x_{ij} - x'_{ij}) = 1 - \sum_i x_{ij}$, nobody receives the item. Note that the resulting scheme is exactly Rounding Procedure 0: agent $i$ receives item $j$ with probability $x_{ij}$, indepenently for each item (since all the randomness that we used was independent for each item $j$). As a result, what agents receive in Rounding Procedure 0 dominates what they receive in Rounding Procedure $2$.
\end{proof}

Let us consider the case of coverage, where each valuation can be viewed as a summation of rank-1 matroid rank functions (corresponding to the coverage of individual elements).
Here, the relaxation looks as follows.

\begin{align*}
(P''_1) \\
&\max_{\bx \in \RR^{n \times m}}  \inf_{\substack{\by \in \RR_+^m:\\ y^S\ge1,\forall S\in\binom{[m]}{n}}} \prod_{i=1}^n&\left( \sum_{e \in U} \sum_{j \in D_{ie}} x_{ije} y_j \right), \\
& s.t.~~~~~~~~~~ \sum_{j \in D_{ie}} x_{ije} \leq 1 & \forall i \in [n], e \in U \\
& ~~~~~~~~~~~~~~ x_{ije} \leq x_{ij} & \forall i \in [n], e \in U, j \in D_{ie} \\
& ~~~~~~~~~~~~~~ \sum_{i=1}^{m} x_{ij}\le 1 & \forall j \in [m] \\
& ~~~~~~~~~~~~~~  x_{ij}, x_{ije} \ge 0 & \forall i \in [n], j \in [m], e \in U \\
\end{align*}

Here, the variable $x_{ije}$ denotes the amount to which item $j$ covers element $e$ for agent $i$; $D_{ie}$ is the set of items whose set for agent $i$ covers element $e$. We consider Rounding Procedure 0, where item $j$ goes to agent $i$ with probability $x_{ij}$. Here, the expected value of the random assignment is
$$ \sum_{\sigma:[m] \rightarrow [n]} \prod_{j=1}^{m} x_{\sigma(j), j} \prod_{i=1}^{n} | \{ e \in U: \exists j \in D_{ie}, \sigma(j) = i\} | $$ 
$$ = \sum_{\sigma:[m] \rightarrow [n]} \prod_{j=1}^{m} x_{\sigma(j), j} \prod_{i=1}^{n} \sum_{e \in U} \b1_{\exists j \in D_{ie}, \sigma(j)=i} $$ 
$$ = \sum_{e_1,\dots,e_n \in U} \sum_{\sigma:[m] \rightarrow[n]} \prod_{j=1}^{m} x_{\sigma(j), j} \prod_{i=1}^{n} \b1_{\exists j \in D_{i e_i}, \sigma(j)=i} $$
$$ = \sum_{e_1,\dots,e_n \in U} \sum_{\sigma: \forall i \exists j \in D_{i e_i}, \sigma(j)=i} \prod_{j=1}^{m} x_{\sigma(j), j}. $$
This last summation can be viewed as the following counting problem: For each agent $i$, fix an element $e_i$ in her coverage universe. Let ${\cal S}(e_1,\ldots,e_n)$ be the set of assignments of items such that each agent $i$ receives some item $j \in D_{i e_i}$, i.e.~some item covering the element $e_i$. We want to count the probability that a random assignment satisfies this property, summed up over all choices of $e_1,\ldots,e_n$. It is conceivable that this counting problem can be reduced to counting of bipartite matchings, but we haven't succeeded in doing so.

\end{document}